\documentclass[11pt,reqno]{amsart}

\usepackage{hyperref}
\usepackage{latexsym,amsmath, bm}
\usepackage{enumerate}
\usepackage{amsfonts}
\usepackage{amssymb}
\usepackage[a4paper,bottom=2.3cm]{geometry}
\usepackage{latexsym}
\usepackage{fixmath}
\usepackage{faktor}
\usepackage{mathtools}

\usepackage[T1]{fontenc}
\usepackage{fourier}
\usepackage{bbm}
\usepackage{color}

\usepackage{pgfplots}
\usepackage{tikz}
\usetikzlibrary{shapes,decorations,arrows,calc,arrows.meta,fit,positioning}
\usepgfplotslibrary{fillbetween}
\usepackage{graphicx}
\usepackage{fullpage}
\usetikzlibrary{decorations.pathreplacing, matrix}

\usepackage[boxed, linesnumbered, noend, noline]{algorithm2e}
\SetKwInput{KwData}{Input}
\SetKwInput{KwResult}{Output}

\numberwithin{equation}{section}

\renewcommand\log{\ln}

\newcommand\SPIV{{\tt SPIV}}

\newcommand\aspiv{{\tt ASPIV}}

\renewcommand{\epsilon}{\eps}

\renewcommand{\vec}[1]{\boldsymbol{#1}}

\newcommand\SIGMA{\vec\sigma}

\newtheorem{definition}{Definition}[section]

\newtheorem{theorem}[definition]{Theorem}
\newtheorem{lemma}[definition]{Lemma}
\newtheorem{proposition}[definition]{Proposition}
\newtheorem{corollary}[definition]{Corollary}

\newcommand\eps{\varepsilon}

\newcommand\Erw{\mathbb{E}}

\newcommand{\set}[1]{\left\{#1\right\}}

\newcommand\bc[1]{\left({#1}\right)}
\newcommand\cbc[1]{\left\{{#1}\right\}}

\newcommand\abs[1]{\left|{#1}\right|}

\newcommand{\Whp}{W.h.p.}
\newcommand{\whp}{w.h.p.}

\newcommand\Lem{Lemma}

\newcommand\Thm{Theorem}

\newcommand\Sec{Section}

 \def\G{{\vec G}}

\newcommand{\remove}[1]{}

\newcommand{\one}{V_1}
\newcommand{\zero}{V_0}
\newcommand{\zerominus}{V_{0-}}
\newcommand{\zeroplus}{V_{0+}}
\newcommand{\oneplus}{V_{1+}}
\newcommand{\oneminus}{V_{1-}}

\newcommand\minf{m_{\mathrm{inf}}}

\newcommand\mSC{m_{\mathrm{SPIV}}}
\newcommand\mnonadap{m_{\mathrm{one-stage}}}

\newcommand{\be}{\begin{equation}}
	\newcommand{\bel}[1]{\begin{equation}\lab{#1}\ }
		\newcommand{\ee}{\end{equation}}
	\newcommand{\bea}{\begin{eqnarray}}
		\newcommand{\eea}{\end{eqnarray}}
	\newcommand{\bean}{\begin{eqnarray*}}
		\newcommand{\eean}{\end{eqnarray*}}

\newcommand\FigG{

\begin{figure}
\begin{minipage}[t][][b]{0.95 \textwidth}
\begin{tikzpicture}[scale=0.8]
\begin{axis}[
axis lines = left,
 xlabel = density parameter $\theta$,
ylabel style={at={(0,0.3)}},
xtick={0, 0.409, 0.5, 1},
xticklabels={$0$, $\frac{\log2}{1+\log2}$, $\frac{1}{2}$, $1$},
ytick={2.081368981, 1.442695041, 0.8520789316},
yticklabels={$\log^{-2}2$,$\log^{-1}2$, $((1+\log2)\log 2)^{-1}$},
ymin = 0,
ymax = 2.2,
xmin = 0,
xmax=1.05,
height=7cm,
width=15cm,
legend style={at={(axis cs:0.5,1.55)},anchor=south}
]

\addplot [
name path=it,
domain=0:1, 
samples=100, 
dotted,
style={ultra thick},
color=red,
]
{max(x/(ln(2)^2), (1-x)/ ln(2)};
\addlegendentry{$m_{\text{one-stage}}/(n^\theta\log n)$}

\addplot [
name path=counting,
domain=0:1, 
samples=100, 
color=blue,
style={ultra thick},
loosely dotted
]
{(1-x)/(ln(2))};
\addlegendentry{$m_{\text{inf}}/(n^\theta\log n)$}

\addplot [
name path=special_designs,
domain=0:1, 
samples=100, 
color=black!80,
style={thick},
loosely dotted
]
{(1-x)/(ln(2)*ln(2))};

\addplot [
        thick,
        color=black,
        fill=red, 
        fill opacity=0.1
    ]
    fill between[
        of=it and counting,
        soft clip={domain=0.409:1},
    ];

\path[name path=axis] (axis cs:0,0) -- (axis cs:1,0);

\addplot+[
mark=none,
color=black,
dotted,
]
coordinates
{(0.409,0) (0.409,0.8520789316)};

\addplot+[
mark=none,
color=black,
dotted,
]
coordinates
{(0,0.8520789316) (0.409,0.8520789316)};

\end{axis}

\end{tikzpicture}
\end{minipage}

\caption{The phase transitions in group testing. The dotted red line shows the information-theoretic and algorithmic phase transition $\mnonadap{}$ for any one-stage test design. Below this line finding the infected individuals is impossible in one-stage group testing, while above the line the efficient inference algorithm \SPIV~exists \cite{Coja_2019_SPIV}. The dotted blue line shows the information-theoretic and algorithmic phase transition $\minf$ for two-stage group testing. The dotted gray line indicates the best previously-known two-stage upper bound \cite{Mezard_2011}. The red area charts the adaptivity gap, that is the regime in which two-stage group testing is superior to one-stage group testing.} \label{fig_bounds_illustration}
\end{figure}

}

\begin{document}
	
\title{Optimal adaptive group testing}
	
\thanks{Supported by DFG CO 646/3 and Stiftung Polytechnische Gesellschaft.}

\author{Max Hahn-Klimroth, Philipp Loick}
\address{Max Hahn-Klimroth, {\tt hahnklim@math.uni-frankfurt.de}, Goethe University, Mathematics Institute, 10 Robert Mayer St, Frankfurt 60325, Germany.}	
\address{Philipp Loick, {\tt loick@math.uni-frankfurt.de}, Goethe University, Mathematics Institute, 10 Robert Mayer St, Frankfurt 60325, Germany.}
	
\begin{abstract}
The group testing problem is concerned with identifying a small number $k \sim n^\theta$ for $\theta \in (0,1)$ of infected individuals in a large population of size $n$. At our disposal is a testing procedure that allows us to test groups of individuals. This paper considers two-stage designs where the test results of the first stage can inform the design of the second stage. We are interested in the minimum number of tests to recover the set of infected individuals w.h.p. Equipped with a novel algorithm for one-stage group testing from [Coja-Oghlan, Gebhard, Hahn-Klimroth \& Loick 2019] and a similar procedure as [Scarlett 2018], we propose a polynomial-time two-stage algorithm that matches the universal information-theoretic lower bound of group testing. This result improves on results from [M\'ezard \& Toninelli 2011] and [Scarlett 2018] and resolves open problems posed in [Aldridge, Johnson \& Scarlett 2019, Berger \& Levenshtein 2002, Damaschke \& Muhammad 2012].
\end{abstract}

\maketitle

\section{Introduction}
\subsection{Background and motivation}
The roots of the group testing problem can be traced back to work of Dorfman in 1943 \cite{Dorfman_1943}. One aims to identify a small number of infected individuals in a large population. At our disposal is a testing procedure that allows us to test groups of individuals. A test will return \textit{positive}, if at least one individual in a tested group is infected, and \textit{negative} otherwise. The merit of this pooled testing procedure is that significantly fewer tests need to be performed than with individual testing. The group testing problem asks for the minimum number of tests under a carefully chosen test design such that reliable recovery of each individual's infection status is possible.

To be precise, let $n$ be the number of individuals and $k \sim n^\theta$ for $\theta \in (0,1)$ be the number of infected individuals. We are interested in the minimum number of tests $m=m(n,\theta)$ from which we can identify the infected individuals. Clearly, the number of possible test results $2^m$ must exceed the $\binom{n}{k}$ possible configurations with $k$ infected individuals. This folklore arguments yields the lower bound 
\begin{align}
    \minf \gtrsim \frac{1-\theta}{\log 2} n^\theta \log n. \label{eq_lower_bound}
\end{align}
In his original contribution, Dorfman proceeded by assigning each individual to exactly one test with each test containing $\sim \sqrt{n/k}$ individuals. If the test result was negative, all individuals in this test could be declared healthy and no further tests were required. Conversely, when a test returned positive, all individuals in this test were subsequently tested individually. In expectation, this test design needed a total of (see \cite{Mezard_2008_2} for the derivation)
\begin{align*}
    \Erw[m_{\text{Dorfman}}] = 2 \sqrt{kn} = 2 n^{(1+\theta)/2}
\end{align*}
tests. Thus, while Dorfman was able to reduce the expected number of tests by a factor of $\Theta(\sqrt{k/n})$ compared to individual testing, the design remained far from optimal. Whether there exists an efficient two-stage algorithm that attains the universal information-theoretic lower bound has been an open problem until today. Particularly over the last years, partial progress has been made in this domain. Scarlett \cite{Scarlett_2019} proposed an efficient algorithm that matches the information-theoretic lower bound, but proceeds in a total of three rather than two stages. Similar results were obtained in \cite{Damaschke_2012}. M\'ezard \& Toninelli \cite{Mezard_2011} showed that a specific class of algorithms fails to recover the infected individuals under any two-stage test design for less than $\minf \log^{-1}2$ tests. In \cite{Scarlett_2018}, Scarlett showed that the lower bound in \eqref{eq_lower_bound} is information-theoretically achievable in two stages.
Finding a \textit{polynomial-time} algorithm that attains the information-theoretic lower bound is still an open research problem \cite{Aldridge_2019, Berger_2002, Damaschke_2012}.
In this paper, we completely resolve this problem by providing a polynomial-time two-stage algorithm that matches the universal lower bound for all $\theta \in (0,1)$.

\FigG{}

The main ingredient to this algorithm is a recent result from \cite{Coja_2019_SPIV} that puts forward an efficient algorithm for \textit{non-adaptive} group testing. Non-adaptive group testing is the one-stage variant of the group testing problem where all tests are performed in parallel and the information from prior tests does not inform the design of subsequent tests. The proposed algorithm called \SPIV~ is inspired by the notion of spatial coupling that has already proved powerful in coding theory \cite{Felstrom_1999, Kudekar_2011, Kudekar_2013}. In \cite{Coja_2019_SPIV}, \SPIV~ is shown to achieve the universal lower bound for all $\theta \leq \log 2 / (1 + \log 2) \approx 0.41$. For $\theta>\log 2 / (1 + \log 2)$, the authors derive an information-theoretic lower bound for any one-stage test design that falls short of the universal lower bound thereby establishing the presence of an adaptivity gap. Specifically, there is no algorithm (efficient or not), that can solve the group testing problem in one stage with less than
$$\mnonadap = \max \cbc{ \frac{1 - \theta}{\log 2}, \frac{\theta}{\log^2 2} } n^\theta \log(n) $$ 
tests. {\tt SPIV} reaches this one-stage information-theoretic lower bound for all $\theta \in (0,1)$, thus it is an optimal one-stage algorithm. Equipped with the result from \cite{Coja_2019_SPIV} and following a similar procedure as \cite{Scarlett_2018}, we find an efficient two-stage algorithm whose bound coincides with the universal lower bound for all $\theta \in (0,1)$. This algorithm will be called \aspiv.  With the main work to derive the \SPIV~algorithm being provided in \cite{Coja_2019_SPIV}, the proofs in this paper are elementary and the key contribution of the present work is the provision of a novel efficient optimal two-stage algorithm. 

\subsection{Results and notation}
Our paper is based on the same assumptions as commonly made in the group testing literature. Their merit will be reviewed in \Sec~\ref{subsec_related_work}. For starters, let $\SIGMA\in\{0,1\}^{\{x_1,\ldots,x_n\}}$ be a vector of Hamming weight $k$ chosen uniformly at random. Throughout this paper, we will refer to $\SIGMA$ as the ground truth. A one-entry in $\SIGMA$ signifies that the corresponding individual is infected. Further, for every stage $i \in \cbc{1,2}$ of the algorithm, let $$\G^{(i)}=\G(n,m^{(i)},(\partial x^{(i)}_j)_{j \in [n]})$$
represent a bipartite factor graph with $n$ vertices \textit{on the left} and $m^{(i)}$ vertices \textit{on the right}. Vertices $x_1, \dots , x_n$ represent the individuals, while $a_1^{(i)}, \dots , a_m^{(i)}$ represent the tests. Two vertices $x$ and $a^{(i)}$ are adjacent, if and only if individual $x$ participates in test $a^{(i)}$. The set of neighbours of vertex $x$ and vertex $a^{(i)}$ under $\G^{(i)}$ will be denoted by $\partial x^{(i)}$ and $\partial a^{(i)}$, respectively. The vertices $x_1,\dots,x_n$ are labeled with values in $0$ and $1$ by $\SIGMA \in \set{0,1}^n$ indicating whether an individual is healthy or infected. A pooling graph $\G$ and $\SIGMA$ induce a vector $\hat\SIGMA(\G)$ indicating the test results. Denote by $\hat\SIGMA^{(i)}\in\{0,1\}^{\{a_1,\ldots,a_m\}}$ the test results w.r.t. $\G^{(i)}$. Specifically, $\hat\SIGMA^{(i)}_a=1$ iff test $a^{(i)}$ is adjacent to an individual $x$ with $\SIGMA_x=1$. We are interested in the minimum number of tests $m$ such that an efficient algorithm can infer $\SIGMA$ from $\hat\SIGMA^{(i)}$ and $\G^{(i)}$ for $i \in \cbc{1,2}$. As we are interested in a two-stage design, we allow $\G^{(2)}$ to depend on $(\G^{(1)}, \hat\SIGMA^{(1)})$. We will further call the tuple of a factor graph $\G$ and a configuration $\SIGMA \in \cbc{0,1}^n$ a \textit{$(\G, \SIGMA)$-group testing instance}. We are in a position to state our main theorem. \footnote{We will write \whp~ for an event that occurs with probability $1-o(1)$.}

\begin{theorem}\label{Thm_aspiv}
Suppose that $0<\theta<1$ and $\eps>0$ and let
\begin{align*}
    \minf = \frac{1-\theta}{\log 2} n^\theta \log n.
\end{align*}
If $m>(1+\eps) \minf$, there exist efficiently constructable test designs $\G^{(1)}, \G^{(2)}$ and a polynomial-time two-stage algorithm that given $\G^{(1)}, \G^{(2)}, \hat\SIGMA^{(1)}, \hat\SIGMA^{(2)}$ outputs $\SIGMA$ \whp
\end{theorem}
\Thm~\ref{Thm_aspiv} immediately implies that the \aspiv~algorithm is an optimal two-stage algorithm since it reaches the universal lower bound \eqref{eq_lower_bound} for all $\theta \in (0,1)$.

\subsection{Related work} \label{subsec_related_work}

In this paper, we focus on two-stage adaptive group testing. Since we reach the universal lower bound \eqref{eq_lower_bound} with our algorithm, there is no need to consider group testing with further stages, since every additional stage is time-consuming and not apt to automation \cite{Chen_2008, Kwang_2006}. Indeed, from an automation perspective, performing all tests in parallel would be the ideal choice. However, a recent result \cite{Coja_2019_SPIV} shows that for $\theta > \frac{\log 2}{1 + \log 2}$, there exists an adaptivity gap resulting in the need of significantly more tests than $\minf$ under any one-stage algorithm. Aldridge \cite{Aldridge_2018} even showed that when $\theta=1$ (aka \textit{linear case}), individual testing is optimal as a one-stage algorithm. Thus, adaptive group testing might be preferable to non-adaptive group testing for large values of $\theta$.

As stated before, we assume throughout the paper that the number of infected individuals $k$ scales sublinearly in the total population $n$, i.e. $k=n^\theta$ for some fixed $\theta \in (0,1)$. This is the setting most commonly considered in the literature \cite{Aldridge_2016,Aldridge_2017, Berger_2002,Coja_2019, Johnson_2019,Mezard_2008,Scarlett_2016}. It also has practical relevance since infections typically spread sublinearly in the total population size as evidenced by Heap's law in epidemiology \cite{Benz_2008} and biological as well as socio-economic effects. For the linear case, it is known that for $k/n > 1-\text{log}_3 2$, individual testing is optimal even when allowing adaptivity \cite{Aldridge_2019}. For $k/n < 1/3$ two-stage algorithms that perform better than individual testing are known \cite{Aldridge_2019_2, Hu_1981, Riccio_2000}.

One further distinction in group testing is between combinatorial and probabilistic group testing. Combinatorial group testing considers all possible $\SIGMA$ and asks for an algorithm that allows to recover $\SIGMA$ in the worst case. Conversely, probabilistic group testing requires the algorithm to recover a configuration $\SIGMA$ chosen uniformly at random \whp~ While combinatorial group testing was the research focus for the decades after the initial proposal by Dorfman, we have witnessed a surge of interest in probabilistic group testing in recent years \cite{Aldridge_2016,Aldridge_2017, Berger_2002,Coja_2019, Johnson_2019,Mezard_2008,Scarlett_2016}, which is also the focus of the paper at hand.

Lastly, the literature often differentiates between combinatorial and iid priors. Under the former, we know that $\SIGMA$ has exactly Hamming weight $k$, while under the latter every individual is infected with probability $k/n$ independently. While the proof of this paper and \cite{Coja_2019_SPIV} pertain to combinatorial priors, Theorem 1.7 of \cite{Aldridge_2019} evinces that our findings immediately transfer to settings with iid priors.

\section{Proof of \Thm~\ref{Thm_aspiv}}
To prove Theorem \ref{Thm_aspiv}, we will need to get a handle on different types of individuals under $\G^{(1)}$. To this end, let
\begin{align*}
    \zero(\SIGMA) = \SIGMA^{-1}(0) \qquad \text{and} \qquad \one(\SIGMA) = \SIGMA^{-1}(1)
\end{align*}
and define
\begin{itemize}
    \item $\zerominus = \zerominus(\G^{(1)},\SIGMA)$: the set of healthy individuals that appear in at least one negative test,
    \item $\zeroplus = \zeroplus(\G^{(1)},\SIGMA)$: the set of healthy individuals that only appear in positive tests, formally
    \begin{align*}
        \zerominus(\G^{(1)}, \SIGMA) = \cbc{x \in \zero(\SIGMA): \min_{a \in \partial_{\G^{(1)}} x} \hat\SIGMA_a^{(1)} = 0 }, \qquad \zeroplus(\G^{(1)}, \SIGMA) = \zero(\SIGMA) \setminus \zerominus(\G^{(1)}, \SIGMA).
    \end{align*}
\end{itemize}
Moreover, let
\begin{itemize}
    \item $\oneminus = \oneminus(\G^{(1)}, \SIGMA)$: the set of infected individuals that appear in at least one test as the only infected individual,
    \item $\oneplus = \oneplus(\G^{(1)}, \SIGMA)$: the set of infected individuals that only appear in tests, that contain at least one other infected individual, formally
    \begin{align*}
        \oneminus(\G^{(1)}, \SIGMA) = \cbc{x \in \one(\SIGMA): \exists a \in \partial_{\G^{(1)}} x: \max_{y \in \partial_{\G^{(1)}} a \setminus x} \SIGMA_y = 0}, \qquad \oneplus(\G^{(1)}, \SIGMA) = \one(\SIGMA) \setminus \oneminus(\G^{(1)}, \SIGMA).
    \end{align*}
\end{itemize}

Detecting individuals in $\zerominus$ is easy by simply classifying all individuals in negative tests under $(\G^{(1)}, \hat\SIGMA^{(1)})$ as healthy. Similarly, it might be within reach to find the set $\oneminus$ by a one-stage algorithm. However, finding individuals in $\zeroplus$ and $\oneplus$ seems a daunting task from the outset since for both of these types we could flip the infection status and obtain an alternative configuration $\tau$ that still leads to the same test results. For this reason, the individuals in $\oneplus$ and $\zeroplus$ are commonly called \textit{disguised} \cite{Aldridge_2018}. It is a highly non-trivial insight that under a specific test-design $\G^{(1)}$ the tuple $(\G^{(1)}, \hat\SIGMA^{(1)})$ contains information that allows us to distinguish between disguised healthy individuals and infected individuals \cite{Coja_2019_SPIV}. 
Specifically, every test in the neighborhood of $x \in \zeroplus$ needs to feature at least one infected individual to render such a test positive, while this is not the case for truly infected individuals. Here, every adjacent test will not contain any other infected individual with probability $1/2$ under a suitable choice of $\G^{(1)}$. We will call such tests which do not contain any further infected individual outside the infected individual under consideration \textit{unexplained}.
The \SPIV~ algorithm exploits this property. On a high level, it proceeds as follows. For the detailed technical statement and proof, we refer the reader to \cite{Coja_2019_SPIV}.

\begin{enumerate}
    \item Infer an estimate $\tau$ of $\SIGMA$ such that $\abs{\abs{ \tau - \SIGMA }}_1 = k n^{-\Omega(1)}$.
    \item Correct the few misclassifications in $\tau$.
\end{enumerate}

\noindent We find the following performance guarantee for \SPIV. Let
\begin{align*}
   \mSC = \mSC(n, k) = \max \left\{ \frac{1 - \theta}{\log 2}, \frac{\theta}{\log^2 2}\right\} n^{\theta} \log(n). 
\end{align*}

\begin{proposition}[Theorem 1.1 of \cite{Coja_2019_SPIV}] \label{Prop_SPIV_tests}Let $k \sim n^\theta$ and $\eps > 0$. The estimate $\tau$ returned by the \SPIV~algorithm will be equal to $\SIGMA$ \whp~ when $m = (1 + \eps) \mSC$.
\end{proposition}

\noindent We will refer to step one of \SPIV~  as the \textit{estimation} phase and to the second step as the \textit{clean-up} phase. The following result from \cite{Coja_2019_SPIV} holds the key to our two-stage algorithm.

\begin{lemma}[Corollary 4.16 in \cite{Coja_2019_SPIV}] \label{fact_two_bounds_spiv}
Let $(\G, \SIGMA)$ be a group testing instance with $k \sim n^\theta$. Furthermore, let $\eps > 0$. If one applies \SPIV~ with $m$ tests, the estimation phase of \SPIV~ leads to an estimation $\tau$ of $\SIGMA$ such that $\abs{\abs{ \tau - \SIGMA }}_1 = k n^{-\Omega(1)}$ \whp, whenever $$m \geq (1 + \eps) (1 - \theta) \log^{-1}(2) n^\theta \log n.$$
\end{lemma}
Therefore, by conducting $m \geq (1 + \eps) \minf$ tests, we can perform the first step of \SPIV. We are now in position to state our two-stage algorithm \aspiv. \\

\subsection*{Stage 1: estimation phase of \SPIV}
As a first step, we apply the estimation phase of \SPIV. By \Lem~\ref{fact_two_bounds_spiv}, employing $m = (1 + \eps) \minf$ tests, we obtain an estimate $\tau$ for which we have \whp 
\begin{align*}
    \abs{\abs{ \tau - \SIGMA }}_1 = k n^{-\Omega(1)}.
\end{align*}

\IncMargin{1em}
\begin{algorithm}[H]
 \KwData{$\hat\SIGMA$, $n, \theta$, $\eps > 0$}
 \KwResult{$\tau$ as an estimate of $\SIGMA$}
   Apply the estimation phase of \SPIV~ with $m^{(1)} = (1 + \eps) \minf$ tests;
\end{algorithm}
\DecMargin{1em}

\noindent Thus, we are left with an estimate of probably infected and probably healthy individuals. Formally,
\begin{align*}
    V_{1}(\tau) = \cbc{x \in V: \tau_x = 1} \qquad \text{and} \qquad V_{0}(\tau) = \cbc{x \in V: \tau_x = 0}.
\end{align*}
At this point, the key insight is, that on the one hand we find very few infected individuals in $\zero(\tau)$ while on the other hand the set $V_{1}(\tau)$ is tiny in comparison to $V_0(\tau)$. We proceed by handling these two sets differently.

\subsection*{Stage 2a: Individual testing for $V_{1}(\tau)$}
We readily obtain an upper-bound on the size of $V_{1}(\tau)$.
\begin{corollary} \label{cor_1}
\Whp,
$$ \abs{V_{1}(\tau)} \leq k + kn^{-\Omega(1)}$$
\end{corollary}

The corollary follows directly from \Lem~\ref{fact_two_bounds_spiv}. On this set of probably infected individuals, we perform individual tests needing $m^{(2)}_1=o(\minf)$ tests \whp

\IncMargin{1em}
\begin{algorithm}[H]
\setcounter{AlgoLine}{1}
   Test all individuals in $V_{1}(\tau)$ individually;
\end{algorithm}
\DecMargin{1em}
\subsection*{Stage 2b: Apply the {\tt COMP} algorithm to $V_{0}(\tau)$}
Let us write $$n'=\abs{V_{0}(\tau)} \leq n \qquad \text{and} \qquad k'= \abs{ \cbc{x \in \zero(\tau): \SIGMA_x = 1}}.$$
By \Lem~\ref{fact_two_bounds_spiv} we immediately find
\begin{corollary}\label{cor_2}
\Whp, we have $k' \leq k/\log n$.
\end{corollary}

$V_{0}(\tau)$ can be regarded as a new group testing instance with $n'$ individuals and $k'$ infected individuals.
Thus, we require a one-stage algorithm that gives a performance guarantee for a group testing instance for which we have an upper bound on the number of infected individuals. The simplest approach, called {\tt COMP}, is well understood, i.e., Aldridge et al. \cite{Aldridge_2016} give a pooling scheme $\G'$, under which one can infer the infection status of all individuals \whp~ by declaring all individuals that appear in negative tests as healthy and declaring all other individuals as infected. Clearly, this leads to the correct classification, iff $\zeroplus(\G',\SIGMA_{\mid V_{0}(\tau)}) = \emptyset.$ The performance guarantee reads
\begin{lemma}[Theorem 2 of \cite{Aldridge_2016}]\label{lem_perf_comb}
Let $k' = O(n'^{\theta'})$ and $\eps > 0$. There is a pooling scheme $\G'$ and a polynomial time algorithm {\tt COMP} that infers $\SIGMA'$ correctly from $\G', \hat\SIGMA'$ as long as
\begin{align*}
    m' > (1+\eps) \frac{1}{\log^2 2} k' \log n'.
\end{align*}
\end{lemma}

\noindent Therefore, \aspiv~ continues as follows.\\ 
\IncMargin{1em}
\begin{algorithm}[H]
\setcounter{AlgoLine}{2}
   Apply the {\tt COMP} algorithm to the set $\zero(\tau)$ with $m_0^{(2)} = (1 + \eps) \frac{1}{\log^2 2}  k$ tests;
\end{algorithm}

\noindent The following lemma evinces that we require $m^{(2)}_0=o(m_1)$ additional tests in step 2b of \aspiv.

\begin{lemma}\label{Lem_2b}
Stage 2b of \aspiv~requires $m^{(2)}_0 = o(m_1)$ tests conducted in parallel to recover $\SIGMA_x$ for every $x \in V_{0}(\tau)$ \whp
\end{lemma}

\begin{proof}
Since $n'\leq n$ and $k'\leq k/\log n$, \Lem~\ref{lem_perf_comb} guarantees that Stage 2b successfully recovers $\SIGMA_x$ for all $x \in V_0(\tau)$, as long as
\begin{align*}
    m_0^{(2)} = (1+ \eps) \frac{1}{\log^2 2} k' \log n' \leq (1 + \eps) \frac{1}{\log^2 2}  k = o\bc{\minf}.
\end{align*}
The lemma follows.
\end{proof}

\begin{proof}[Proof of Theorem \ref{Thm_aspiv}]
Observe that Stage 2a and Stage 2b can be performed in parallel such that \aspiv~ is indeed a two-stage algorithm. Moreover, by \Lem~\ref{fact_two_bounds_spiv}, Corollaries \ref{cor_1}, \ref{cor_2} and \Lem~\ref{Lem_2b} \aspiv~requires
\begin{align*}
    m = m^{(1)} + m^{(2)}_0 + m^{(2)}_1 \leq (1+ 3 \eps) \minf
\end{align*}
tests to recover $\SIGMA$ \whp~ The theorem follows.
\end{proof}

\section{Concluding Remarks} \label{sec_related_work}
For probabilistic sublinear group testing - the setting we consider here - M\'ezard \& Toninelli \cite{Mezard_2011} analysed a specific class of algorithms for which they find that when
\begin{align}
  m<(1-\eps) \frac{\minf}{\log 2} \label{eq_Mezard}  
\end{align}
reliable recovery of $\SIGMA$ is not possible under any two-stage test design. The kind of algorithms are such that they are able to identify $\zerominus$ and $\oneminus$ in the first stage and then perform individual testing on the sets $\zeroplus$ and $\oneplus$ in the second stage. The key argument of \cite{Mezard_2011} is that as long as \eqref{eq_Mezard} holds $\abs{\zeroplus}, \abs{\oneplus} = O(k)$ and thus the number of tests performed in the second stage of the algorithm are of lower order than $\minf$. Indeed, until the recent work of \cite{Coja_2019_SPIV}, no algorithm was known that could identify individuals $x \in \zeroplus$ when \eqref{eq_Mezard} holds. 
It was tempting to assume that no such algorithm exists. 

By definition, an individual $x \in \zeroplus$ only shows up in positive tests. Thus, one could easily find a second configuration $\tau$ by flipping this individual from healthy to infected and yield the same test results $\hat \SIGMA$. Without further information, there is no way to tell $\SIGMA$ apart from $\tau$ and the best option is to uniformly at random select one configuration. However, from a Bayesian perspective, this assumption is not true. To be precise, there is more information hidden in the graph which allows us to tell the sets $\zeroplus$ and $\one$ apart.
First, even though $\tau$ is a valid alternative configuration satisfying the test results, the prior probability for an individual to be infected informs us that $\SIGMA$ is far more likely than $\tau$. Second, when considering the neighbourhood of any individual $x \in \zeroplus$, one will find at least one infected individual rendering this test positive. Conversely, for $y \in \one$ the probability to find at least one other infected individual is $\sim 1/2$ under the (best possible) choice of test design where half of the tests are expected to be positive \cite{Aldridge_2016, Aldridge_2017, Coja_2019_SPIV}. Thus, the posterior distribution on configurations yielding the same test results is highly biased towards the correct configuration $\SIGMA$. This insight in conjunction with the notion of spatial coupling was used to derive an algorithm that already in one stage correctly classifies all but a polynomially small fraction of all individuals in $\zeroplus$ and $\one$ by employing $\minf$ tests \cite{Coja_2019_SPIV}. This property is the key which allows the \aspiv~ algorithm to achieve a bound beyond \eqref{eq_Mezard}. This finding thus improves on the work of \cite{Mezard_2011} and resolves open problems posed in \cite{Aldridge_2019, Berger_2002, Damaschke_2012, Scarlett_2018}. 

Most importantly, while it was shown in \cite{Coja_2019_SPIV} that there exists an adaptivity gap for one-stage test designs, our result evinces that already with two stages the universal information-theoretic lower bound is attainable by an efficient algorithm.

%\begin{remark}\label{remark_iid}
%For the \SPIV~ algorithm to work, we do not need knowledge about the specific value of $k$, but the order $k \sim n^{\theta}$ is sufficient. Similarly, Stages 2a and 2b of \aspiv~ get by with an estimate of $k$ rather than the exact value. Therefore, \Thm~\ref{Thm_aspiv} immediately extends from combinatorial to iid priors.
%\end{remark}

\subsection*{Acknowledgment}
We thank Amin Coja-Oghlan, Oliver Gebhard and Maurice Rolvien for fruitful discussions on the performance of the \SPIV~ algorithm.

\end{document}